\documentclass{article}
\usepackage{amsfonts}
\usepackage{graphicx}
\usepackage{amsmath}
\usepackage{hyperref}

\newtheorem{theorem}{Theorem}

\newtheorem{corollary}[theorem]{Corollary}

\newtheorem{definition}[theorem]{Definition}

\newtheorem{lemma}[theorem]{Lemma}

\newtheorem{proposition}[theorem]{Proposition}
\newtheorem{remark}[theorem]{Remark}

\newenvironment{proof}[1][Proof]{\textbf{#1.} }{\ \rule{0.5em}{0.5em}}

\newcommand{\func}{\operatorname}
\newcommand{\limfunc}{\operatorname}

\begin{document}

\title{Poisson brackets of even symplectic forms on the algebra of differential forms }
\author{J. Monterde and J. A. Vallejo\\
Departament de Geometria i Topologia.\\
Universitat de Val\`{e}ncia (Spain).\\
{\small{e-mail: \texttt{juan.l.monterde@uv.es, jose.a.vallejo@uv.es}}}
}
\maketitle

\begin{abstract}
Given a symplectic form and a pseudo-riemannian metric on a manifold, a non
degenerate even Poisson bracket on the algebra of differential forms is defined and its
properties are studied. A comparison with the Koszul-Schouten bracket is established.
\end{abstract}

\section{Introduction}

The extension of a classical Poisson bracket to a bracket defined on
differential forms is a topic studied by many authors (\cite{[4]}, \cite{[7]}%
, \cite{[8]}, \cite{[9]}). Let us consider the problems posed by this
extension.

Let $M$ be a differentiable manifold of dimension $n$, and let $\Omega
(M)=\bigoplus_{k=0}^{n}\Omega^{k}(M)$ be the algebra of differential forms
on $M$. We want to consider Poisson brackets on $\Omega(M)$ and, as this
space has a natural grading, we need to generalize the usual definition to
this graded setting. Thus, a graded Poisson bracket on $\Omega(M)$ of $%
\mathbb{Z}-$degree $k\in\mathbb{Z}$, will be a mapping $[\![\_,\_]\!]:\Omega
(M)\times\Omega(M)\longrightarrow\Omega(M)$ verifying the conditions (here $%
\alpha,\beta,\gamma\in\Omega(M)$ and $\left| \_\right| $ denotes the degree
in $\Omega(M)$):

\begin{enumerate}
\item  $\mathbb{R}-$bilinearity

\item  $\left| [\![\alpha,\beta]\!]\right| =\left| \alpha\right| +\left|
\beta\right| +k$

\item  $[\![\alpha,\beta]\!]=-(-1)^{(\alpha+k)(\beta+k)}[\![\beta,\alpha]\!]$
(graded conmutativity)

\item  $[\![\alpha,\beta\wedge\gamma]\!]=[\![\alpha,\beta]\!]\wedge
\gamma+(-1)^{(\alpha+k)\beta}\beta\wedge\lbrack\![\alpha,\gamma]\!]$
(Leibniz rule)

\item  $[\![\alpha,[\![\beta,\gamma]\!]]\!]=[\![[\![\alpha,\beta]\!],\gamma
]\!]+(-1)^{(\alpha+k)(\beta+k)}[\![\beta,[\![\alpha,\gamma]\!]]\!]$ (graded
Jacobi identity).
\end{enumerate}

So, we see that a Poisson bracket on $\Omega(M)$ can be even or odd,
depending on $k\in\mathbb{Z}$. The best known brackets are those of odd
degree; a first step towards this process of passing from functions to
differential forms can be found in \cite{[1]} where a definition of a
bracket on differential $1$-forms satisfying a Jacobi identity is presented.
The extension of this bracket to another defined on the whole algebra of
differential forms (which we call here the Koszul-Schouten bracket and
denote $[\![\_,\_]\!]_{KS}$) was given in \cite{[6]}, with odd degree, and
has been used and generalized ever since in many ways. Its simple definition
and its good behavior with respect to the exterior derivative make it a nice
candidate to be used for applications in Physics and Geometry.

However, if we assume a Poisson bracket $\{\_,\_\}$ defined on $M$, the odd
brackets such as the Koszul-Schouten, are not objects that could be called
extensions of $\ \{\_,\_\}$, as they do not satisfy 
$$
\pi_{(0)}([\![f,h]\!])=\{f,h\},
$$
where $\pi_{(0)}$ denotes the $0$-degree term.

Our purpose in this paper is to look for a definition of an even graded
Poisson bracket on $\Omega (M)$ which satisfies the preceding condition,
with the additional requirement that the symplectic form associated to the
bracket should be as simple as possible, allowing the explicit computations
of such things as Hamiltonian vector fields.

As it is well known, the Jacobi identity is a highly nonlinear condition and
this fact makes it hard to find graded Poisson brackets. In the non
degenerate case, one can study symplectic forms instead of Poisson brackets.
The Jacobi condition is then translated into a linear condition: a non
degenerate $2$-form is symplectic if it is closed. This is what Rothstein
did in \cite{[13]} where he studied even symplectic forms on any graded
manifold, and where that condition is always verified.

In this paper \cite{[13]}, the author states that any even symplectic form
on a graded manifold $(M,\Gamma(\Lambda E))$ is the pull-back, by an
automorphism of $\Gamma(\Lambda E)$, of a symplectic form of $\mathbb{Z}-$%
degree $(0)+(2)$. Moreover, this $(0)+(2)-$symplectic form can be
parametrized by three objects: a symplectic form $\omega$ on the base
manifold $M$; a non degenerate bilinear form $g$, on the fiber bundle $%
E^{\ast}$, and a connection, $\nabla$, on $E$.

We shall be dealing with the graded manifold \ defined by the cotangent
bundle; that is $E=T^{\ast}M$, so that the differential forms on $M$ are
graded functions. Now, since the space of connections has an affine
structure, and since a non degenerate bilinear form $g$ on $TM$ is a
pseudoriemannian metric on $M$ which has a canonically associated
connection, its Levi-Civita connection, we can replace $\nabla$ by a tensor
field $L\in\Gamma(T^{\ast }M\otimes\Lambda^{2}T^{\ast}M)$. Therefore, a
class of even graded symplectic forms can be parametrized with these three
objects on the base manifold: a symplectic form, a pseudoriemannian metric
and the tensor field $L$.

In the odd case, a result similar to that of Rothstein particularized to the
graded manifold $(M,\Omega(M))$ (see \cite{[2]}), states that any odd
symplectic form on $(M,\Omega(M))$ is the pull-back by an automorphism $%
\varphi\in\limfunc{Aut}\Omega(M)$ of a symplectic form of $\mathbb{Z}-$%
degree $1$. Moreover, this symplectic form of $\mathbb{Z}-$degree $1$ can be
parametrized just by one object: a bundle isomorphism, $P$, between tangent
and cotangent bundles of the base manifold. The odd symplectic form is a
graded exact form $\func{d}^{G}\lambda_{P}$ (in a sense that will be
explained in the next section), where $\lambda_{P}$ is a graded $1-$form
associated to the isomorphism $P$ and $\func{d}^{G}$ is the graded exterior
derivative (see Section $2$ for the definitions).

Among all the derivations of the algebra of differential forms, i.e, graded
vector fields of the graded manifold $(M,\Omega(M))$, there is a
distinguished one: the exterior derivative, $d$. It is natural to ask for
the odd graded symplectic forms that admits it as a Hamiltonian vector
field. In \cite{[2]}, it is proved that these forms are precisely those
which come from a bundle isomorphism between tangent and cotangent bundles
induced by a symplectic form, $\omega$, on the base manifold $M$. Moreover,
in this case $d$ is the Hamiltonian vector field corresponding to $\omega$
(denoted $D_{\omega }=\func{d}$), considered as a graded function. The odd
Poisson bracket associated to $\func{d}^{G}\lambda_{\omega}$ is exactly the
Koszul-Schouten bracket \cite{[6]}.

In the same paper \cite{[2]}, it is shown that $\func{d}$ can never be a
Hamiltonian vector field for an even symplectic form. But we can ask if
replacing $\func{d}$ for another derivation allows us to maintain the
analogy with the odd case. As we will see the answer is in the affirmative
with the derivation $i_{J}$ playing the r\^{o}le of $\func{d}$, i.e, the
insertion of the tensor field $J$ defined by $\omega(X,Y)=g(JX,Y)$.

All this can be gathered together in the following table:

\bigskip

\begin{tabular}{c|c|c|}
\cline{2-2}\cline{2-3}
& Odd & Even \\ \hline
\multicolumn{1}{|c|}{Obtained as} & $\Theta=\varphi^{\ast}(\Theta_{P})$ & $%
\Theta=\varphi^{\ast}(\Theta_{\omega,g,L})$ \\ 
\multicolumn{1}{|c|}{($\varphi\in\limfunc{Aut}\Omega(M)$)} &  &  \\ \hline
\multicolumn{1}{|c|}{Characterizing} & $\func{d}\in \mathrm{\limfunc{Ham}}%
(\Theta_{P})$ & $i_{J}\in \mathrm{\limfunc{Ham}}(\Theta_{\omega,g,L})$ \\ 
\multicolumn{1}{|c|}{condition} & $\Longleftrightarrow P\text{ induced by }%
\omega$ & $\Longleftrightarrow i_{J}L=0$ \\ \hline
\multicolumn{1}{|c|}{Hamiltonian vector} & $\func{d}$ & $i_{J}$ \\ 
\multicolumn{1}{|c|}{field for $\omega$} &  &  \\ \hline
\end{tabular}

\bigskip

Moreover, a subclass of even symplectic forms can be selected by its special
properties which make them easier to work with. This subclass is obtained by
demanding $L=0$, or, equivalently, that the insertion of $\func{d}$ on the
even symplectic form be equal to $\lambda_{\omega}$. Let us remark that this
subclass is defined just through a symplectic form, $\omega$, and a
pseudoriemannian metric both on the manifold $M$. Just with this two
ingredients, and with no condition between them, it is possible to define a
true extension of the classical Poisson bracket associated to $\omega$.

There are many situations where this two ingredients are present: K\"ahler
manifolds, para-complex geometry, Lagrangian problems defined through a
Riemannian metric on a manifold, \dots Our belief is that, because of the
naturality of the even symplectic form and nice relations with the
Koszul-Schouten bracket, the associated extended even Poisson bracket will
find its applications in such kind of situations.

In this work we make a detailed study of the properties of this class of
even symplectic manifolds along with those of the corresponding induced
graded Poisson structure. One of the results we obtain (see Theorem $1$),
states a certain ``integrability property'' of the Koszul-Schouten (odd
graded) bracket and the even graded Poisson bracket we define, this property
being expressed in terms of the graded symplectic forms which they induce.
Also, we give some examples and applications related to complex geometry.

\section{Graded forms on $\Omega(M)$}

Let $\Omega(M;TM)=\bigoplus_{k=0}^{n}\Omega^{k}(M;TM)$ be the graded left $%
\Omega(M)$-module of the vector-valued differential forms. We adopt the
convention that if $v$ is an element of a graded module and the notation $%
\left| v\right| $ is used, we are tacitly assuming that $v$ is homogeneous
of degree $\left| v\right| $. The graded left module $\Omega(M;TM)$ can also
be viewed as a graded right $\Omega(M)$-module with the multiplication $%
S\wedge\alpha=(-1)^{\left| S\right| \left| \alpha\right| }\alpha\wedge S$,
for $\alpha\in\Omega(M)$ and $S\in\Omega(M;TM)$. Let ${\mathrm{\limfunc{Der}}%
}\Omega(M)$ be the graded right $\Omega (M)$-module of all derivations on $%
\Omega(M)$. ${\mathrm{\limfunc{Der}}}\Omega(M)$ is a graded Lie algebra with
the usual graded commutator. Unless otherwise stated, \textit{linear} will
mean ${}\mathbb{R}-$linear.

The natural grading of the algebra $\Omega(M)$ is the ${}\mathbb{Z}-$%
grading, but sometimes we will also refer to the $\mathbb{Z}_{2}$-grading.
In this case, homogeneous elements, or any other homogeneous structures,
will be called even, if $k=0$, or odd if $k=1$.

Recall that the elements of the graded module of derivations, ${\mathrm{%
\limfunc{Der}}}\Omega(M)$, can be regarded as graded vector fields on the
graded manifold based on $\Omega(M)$. By analogy, a graded differential form
is an $\Omega(M)$-multilinear alternating graded homomorphism from the
module of graded vector fields into $\Omega(M)$ (We shall refer to \cite{[5]}
for definitions).

Being a graded homomorphism of graded modules, a graded differential form
has a degree. Thus, we can define a $\mathbb{Z}\times\mathbb{Z}-$bigrading
on the module of graded differential forms and we will say that a graded
differential form $\lambda$ has bidegree $(p,k)\in\mathbb{Z}\times\mathbb{Z}$
if 
\begin{equation*}
\lambda:{\mathrm{\limfunc{Der}}}\Omega(M)\times.\overset{p)}{.}.\times{%
\mathrm{\limfunc{Der}}}\Omega(M)\longrightarrow\Omega(M)
\end{equation*}
and if, for all $D_{1},...,D_{p}\in$ ${\mathrm{\limfunc{Der}}}\Omega (M)$, 
\begin{equation*}
\left| \left\langle D_{1},...,D_{p};\lambda\right\rangle \right|
=\sum\limits_{i=1}^{p}\left| D_{i}\right| +k\text{.}
\end{equation*}

\begin{remark}
In Physics literature, $p$ is called the cohomological degree.
\end{remark}

Using this bigrading, any graded $p-$differential form $\lambda$ can be
decomposed as a sum $\lambda=\lambda_{(0)}+...+\lambda_{(n)}$, where $%
\lambda_{(i)}$ is an homogeneous graded form of bidegree $(p,i)$.

The insertion operator is defined as usual. If $\lambda$ is a graded form of
bidegree $(p,k)$ and $D$ is a derivation of degree $\left| D\right| $, then $%
\iota_{D}\lambda$ (or $\iota(D)\lambda$) is the graded form of bidegree $%
(p-1,k+\left| D\right| )$ defined by 
\begin{equation*}
\langle D_{1},\dots,D_{p-1};\iota_{D}\lambda\rangle:=\langle
D_{1},\dots,D_{p-1},D;\lambda\rangle.
\end{equation*}

This implies that the bidegree of the operator $\iota(D)$ is $(-1,|D|)$.

We shall denote by $\func{d}^{G}$ the graded exterior differential. (See 
\cite{[5]} for details.) In particular, for a graded $0$-form $\alpha
\in\Omega(M)$, $\langle D;\func{d}^{G}\alpha\rangle=D(\alpha)$, and for a
graded $1$-form, $\lambda$, on $\Omega(M)$ we have 
\begin{equation*}
\langle D_{1},D_{2};(\func{d}^{G}\lambda)\rangle=D_{1}(\langle
D_{2};\lambda\rangle)-(-1)^{\left| D_{1}\right| \left| D_{2}\right|
}D_{2}(\langle D_{1};\lambda\rangle)-\langle\lbrack D_{1},D_{2}];\lambda
\rangle.
\end{equation*}
The graded exterior differential is an operator of bidegree $(1,0)$.

A fundamental result is the following corollary to a theorem by Kostant (4.7
in \cite{[5]}).

\begin{proposition}
\label{corKos} Every $\func{d}^{G}$-closed graded form of bidegree $(p,k)$
with $k>0 $ is exact.
\end{proposition}

Other familiar operators on ordinary manifolds also have counterparts on
graded manifolds. If $D\in{\mathrm{\limfunc{Der}}}\Omega(M)$, then the Lie
operator $\mathcal{L}_{D}^{G}$ is defined by 
\begin{equation*}
\mathcal{L}_{D}^{G}=\iota(D)\circ\func{d}^{G}+\func{d}^{G}\circ\iota(D).
\end{equation*}

Note that $\mathcal{L}_{D}^{G}$ is an operator of bidegree $(0,|D|)$. We
also have some formulae similar to the classic ones, as $\lbrack\mathcal{L}%
_{D}^{G},\func{d}^{G}]=0$, or $\lbrack\mathcal{L}_{D}^{G},\iota_{E}]=\iota_{%
\lbrack D,E]}. $

Now, we extend to the graded case some important concepts of Hamiltonian
(symplectic) mechanics.

\begin{definition}
Let $\Theta$ be a graded symplectic form of degree $k\in\mathbb{Z}_{2}(%
\mathbb{Z})$ on $\Omega(M)$, i.e. a $\func{d}^{G}-$closed, non degenerate
graded $2-$form.

\begin{enumerate}
\item  A graded vector field $D$ is a Hamiltonian graded vector field if the
graded $1$-form $\iota(D)\Theta$ is $\func{d}^{G}$-exact; i.e, if there
exits $\alpha\in\Omega(M)$, such that $\iota(D)\Theta=\func{d}^{G}\alpha$.
We shall denote a Hamiltonian $D$ by $D_{\alpha}$ in order to emphasize the
differential form on $M$ to which $D$ is associated, and sometimes even $%
D_{\alpha}^{\Theta}$.

\item  $D$ is a locally Hamiltonian graded vector field if the graded $1$%
-form $\iota(D)\Theta$ is $\func{d}^{G}$-closed. Equivalently, $\mathcal{L}%
_{D}^{G}\Theta=0$.
\end{enumerate}
\end{definition}

\begin{remark}
As a consequence of corollary $2.1$, every locally Hamiltonian graded vector
field of positive degree is a Hamiltonian graded vector field. \medskip
\end{remark}

Note that ${\mathrm{\limfunc{Der}}}\Omega(M)$ is a graded locally free $%
\Omega(M)-$module (this is a consequence of a classical theorem of
Fr\"{o}licher-Nijenhuis), so any graded form is uniquely determined by its
action on derivations $\{i_{X},\mathcal{L}_{X}\}$, where $i_{X}$ and $%
\mathcal{L}_{X}$ are the insertion operator and the Lie derivative with
respect to a vector field $X$ on $M$, respectively.

\section{Definition of a class of even symplectic forms}

As it has been mentioned in the introduction, an even symplectic form on a
graded manifold can be modeled by three objects: an usual symplectic form on
the base manifold, a pseudoriemannian metric on the characteristic bundle,
and a connection in the fiber bundle.

In the particular case of the graded manifold $(M,\Omega(M))$, it is
possible to define even symplectic forms without the need of any linear
connection. We shall start with a symplectic form $\omega$, on a
differentiable manifold $M$, together with a pseudoriemannian metric on $M$, 
$g$.

Let us define the graded $2$-form $\Theta_{\omega}$, on $(M,\Omega(M)),$ as
follows: for any pair of derivations on $\Omega(M)$, $D_{1},D_{2}$ 
\begin{equation*}
<D_{1},D_{2};\Theta_{\omega}>:=\omega(\tilde{D}_{1},\tilde{D}_{2}),
\end{equation*}
where $\tilde{D}$ is the unique vector field on $M$ defined by $\tilde
{D}(f)=\pi_{0}(D(f))$, for any $f\in C^{\infty}(M)$ (here $\pi_{0}$ is the
projection assigning to each differential form its $0-$degree component).
For example, $\left\langle \mathcal{L}_{X},\mathcal{L}_{Y};\Theta_{\omega
}\right\rangle =\omega(X,Y).$

According to the $\mathbb{Z}-$bigraduation of the module of graded
symplectic forms, $\Theta_{\omega}$ has bidegree $(2,0)$ and it is
degenerate. Indeed, we have $\left\langle
i_{X},D;\Theta_{\omega}\right\rangle =0$ for any derivation $D$. So, in
order to define a graded symplectic form we need to add another term.

To introduce this new term, recall that as a consequence of the
Fr\"{o}licher-Nijenhuis theorem, if $\{X_{1},X_{2},\dots,X_{n}\}$ is a local
basis of vector fields, then $\{\mathcal{L}_{X_{1}},...,\mathcal{L}%
_{X_{n}},i_{X_{1}},...,i_{X_{n}}\}$ is a graded basis of the $\Omega(M)$%
-graded module of derivations.

On the other hand, given the pseudoriemannian metric $g$ on $M$, we can
define the graded $1$-form $\lambda_{g}$ as follows. Given any vector field $%
X$ we have 
\begin{gather}
<i_{X};\lambda_{g}>=g(X,\_)=\flat(X)  \label{eq3.1} \\
\left\langle \mathcal{L}_{X};\lambda_{g}\right\rangle =\func{d}(g(X,\_))=%
\func{d}\flat(X).  \notag
\end{gather}
where $\flat(X)$ denotes the $1-$form canonically associated to a vector
field $X$ ``lowering indexes'' with the metric $g$.

For an arbitrary derivation $D$, we extend this definition by $\Omega (M)$%
-linearity. According to the $\mathbb{Z}-$bigraduation of the module of
graded symplectic forms, $\lambda_{g}$ has bidegree $(1,2).$

\begin{proposition}
The graded $2$ form $\Theta_{\omega,g}:=\Theta_{\omega}+\frac{1}{2}\func{d}%
^{G}\lambda_{g}$ is an even symplectic form on $(M,\Omega(M))$.
\end{proposition}

\begin{proof}
Simply note that $\Theta _{\omega }$ is a closed form as we can see after
some computations. In fact, 
\begin{equation*}
\left\langle D_{1},D_{2},D_{3};\func{d}^{G}\Theta _{\omega }\right\rangle =(%
\func{d}\omega )(\tilde{D}_{1},\tilde{D}_{2},\tilde{D}_{3})=0,
\end{equation*}
because $\omega $ is a usual symplectic form. The non degeneracy follows
immediately from the next result.
\end{proof}

\begin{proposition}
\label{prop3} The expression of the even symplectic form $\Theta_{\omega,g}$
with respect to the local basis is given by the formulae 
\begin{align*}
\left\langle \mathcal{L}_{X},\mathcal{L}_{Y};\Theta_{\omega,g}\right\rangle
& =\omega(X,Y)+\alpha(X,Y)\text{ }where\text{ }\alpha(X,Y)\in\Omega ^{2}(M)%
\text{ and} \\
& (\alpha(X,Y))(U,V)=g(\nabla_{U}Y,\nabla_{V}X)-g(\nabla_{U}X,\nabla
_{V}Y)-R(X,Y,U,V), \\
\left\langle \mathcal{L}_{X},i_{Y};\Theta_{\omega,g}\right\rangle &
=g(\nabla_{\_}X,Y), \\
\left\langle i_{X},i_{Y};\Theta_{\omega,g}\right\rangle & =g(X,Y),
\end{align*}
where $X,Y,U,V\in\mathcal{X}(M)$, $\nabla$ is the Levi-Civita connection of $%
g$ and $R$ its Riemann curvature tensor.
\end{proposition}

\begin{proof}
It is just a matter of computation using the definitions. Note that these
expressions imply 
\begin{equation*}
\det\tilde{\Omega}_{\omega,g}=\det\omega\det g\neq0,
\end{equation*}
where $\tilde{\Omega}_{\omega,g}$ is the matrix built with the $0-$degree
terms of the entries of $\Omega_{\omega,g}$, and that the condition for an
arbitrary graded symplectic form $\Lambda$ to be non-degenerated, is
precisely that its associated $\tilde{\Lambda}$ has $\det\tilde{\Lambda}%
\neq0 $ (see \cite{[5]} for details).
\end{proof}

\begin{remark}
\label{formulanabla} Alternatively, we could use $\{\nabla_{X_{1}},...,%
\nabla_{X_{n}},i_{X_{1}},...,i_{X_{n}}\}$ as a local basis of derivations.
Then, a straightforward computation using the relation $\mathcal{L}%
_{X}=\nabla_{X}+i_{\nabla\_X}$ gives 
\begin{align*}
\left\langle \nabla_{X},\nabla_{Y};\Theta_{\omega,g}\right\rangle &
=\omega(X,Y)-g(R(X,Y)\_,\_), \\
\left\langle \nabla_{X},i_{Y};\Theta_{\omega,g}\right\rangle & =0, \\
\left\langle i_{X},i_{Y};\Theta_{\omega,g}\right\rangle & =g(X,Y).
\end{align*}
\end{remark}

The family of even symplectic forms thus defined $\Theta_{\omega,g}$ is just
a subclass of even symplectic forms of bidegree $(2,0+2)$. Now, we give some
indications about the general case (see \cite{[10]} for details).

The most general even graded symplectic form $\Theta$, can be written as 
\begin{equation*}
\Theta=\Theta_{\omega}+\frac{1}{2}\func{d}^{G}\lambda_{g,L}:=\Theta_{%
\omega,g,L}
\end{equation*}
where $g,\omega$ are a metric and a symplectic form on the base manifold $M$
(as before) respectively, $\Theta_{\omega}$ has been already defined and $%
L\in T^{(1,2)}(M)$. The graded $1-$form $\lambda_{g,L}$ is defined by 
\begin{align*}
\left\langle i_{X};\lambda_{g,L}\right\rangle & =g(X,\_)=\flat(X), \\
\left\langle \mathcal{L}_{X};\lambda_{g,L}\right\rangle & =\func{d}%
\flat(X)+L(X;\_,\_).
\end{align*}

Let us also recall how to construct an odd graded symplectic form,
concretely the Koszul-Schouten form (see \cite{[6]}). For this, we only need
a usual symplectic form on the base manifold, $\omega$, and to define the
graded $1-$form $\lambda_{\omega}$ by 
\begin{align}
\left\langle i_{X};\lambda_{\omega}\right\rangle & =0  \label{deflamdaomega}
\\
\left\langle \mathcal{L}_{X};\lambda_{\omega}\right\rangle & =\omega (X,\_).
\notag
\end{align}

Note that while $\lambda_{g}$ has bidegree $(1,2)$, $\lambda_{\omega}$ has
bidegree $(1,-1)$. Now, the Koszul-Schouten form $\Theta_{KS}$ is the odd
graded exact form given by 
\begin{equation*}
\Theta_{KS}=\limfunc{d}\nolimits^{G}\lambda_{\omega}\text{.}
\end{equation*}

\section{Relationship with the Koszul-Schouten bracket and characterization
of the class $\Theta_{\protect\omega,g}$}

As we have just seen, Koszul has given a way to construct, from a symplectic
manifold $(M;\omega)$, an odd graded symplectic structure on the graded
manifold $(M,\Omega(M))$. The corresponding Poisson bracket, called the
Koszul-Schouten bracket, is extremely simple, and this situation is in
strong contrast with that of even graded symplectic forms, which are harder
to handle.

The class of graded forms of type $\Theta_{\omega,g}$, enjoy several
properties that make them relatively easier to work with. We can understand
these properties by showing an intimate relation with the Koszul-Schouten
form, which is characteristic of this class.

\begin{lemma}
Let $\lambda_{g}$ be defined as in Section $3$. Then 
\begin{equation*}
\left\langle \func{d};\lambda_{g}\right\rangle =0.
\end{equation*}
\end{lemma}

\begin{proof}
It is based on a coordinate computation, using the fact that $\limfunc{d}=%
\mathcal{L}_{Id}$. Let $A,B,C\in\mathcal{X}(M)$, then: 
\begin{align*}
\left\langle \limfunc{d};\lambda_{g}\right\rangle (A,B,C) & =(\limfunc{d}%
x^{k}\wedge\left\langle \mathcal{L}_{\partial_{k}};\lambda_{g}\right\rangle
)(A,B,C)= \\
& =(\limfunc{d}x^{k}\wedge\limfunc{d}g(\partial_{k},\_))(A,B,C)= \\
& =A^{k}.(Bg(\partial_{k},C)-Cg(\partial_{k},B))- \\
& B^{k}.(Ag(\partial_{k},C)-Cg(\partial_{k},A))+ \\
& C^{k}.(Ag(\partial_{k},B)-Bg(\partial_{k},A))- \\
&
A^{k}.g(\partial_{k},[B,C])+B^{k}.g(\partial_{k},[A,C])-C^{k}.g(%
\partial_{k},[A,B])- \\
& g(A,[B,C])+g(B,[A,C])-g(C,[A,B])=0.
\end{align*}
\end{proof}

\begin{theorem}
With the preceding notations, $\iota_{\func{d}}\Theta_{\omega
,g}=\lambda_{\omega}$ (where $\lambda_{\omega}$ is given by (\ref
{deflamdaomega})) so 
\begin{equation*}
\mathcal{L}_{\func{d}}^{G}\Theta_{\omega,g}=\Theta_{KS}.
\end{equation*}
\end{theorem}

\begin{proof}
If we consider $\iota_{\func{d}}\Theta_{\omega,g}$ acting on basic
derivations, and make repeated use of Lemma $8$, we find ($Y\in\mathcal{X}%
(M) $) on $\mathcal{L}_{Y}$: 
\begin{align*}
\left\langle \mathcal{L}_{Y},\limfunc{d};\Theta_{\omega,g}\right\rangle &
=\left\langle \mathcal{L}_{Y},\limfunc{d};\Theta_{\omega }\right\rangle +%
\frac{1}{2}\left\langle \mathcal{L}_{Y},\limfunc{d};\limfunc{d}%
\nolimits^{G}\lambda_{g}\right\rangle =\omega(Y,\_)-\frac {1}{2}\left\langle 
\mathcal{L}_{Y};\lambda_{g}\right\rangle \\
& =\omega(Y,\_)-\frac{1}{2}\limfunc{d}\limfunc{d}Y_{\flat}=\omega(Y,\_)=%
\left\langle \mathcal{L}_{Y};\lambda_{\omega}\right\rangle ,
\end{align*}
and on $\iota_{Y}$: 
\begin{align*}
2\left\langle \iota_{Y},\limfunc{d};\Theta_{\omega,g}\right\rangle &
=\left\langle \iota_{Y},\limfunc{d};\limfunc{d}\nolimits^{G}\lambda_{g}%
\right\rangle =\limfunc{d}\left\langle \iota_{Y};\lambda _{g}\right\rangle
-\left\langle [\iota_{Y},\limfunc{d}];\lambda _{g}\right\rangle \\
& =\limfunc{d}Y_{\flat}-\left\langle \mathcal{L}_{Y};\lambda
_{g}\right\rangle =\limfunc{d}Y_{\flat}-\limfunc{d}Y_{\flat
}=0=2\left\langle \iota_{Y};\lambda_{\omega}\right\rangle .
\end{align*}
\end{proof}

\begin{remark}
J. Grabowski (see \cite{[4]}), has constructed an even bracket on
differential forms which is a kind of ``integral'' of the Koszul-Schouten
one, in the sense that (if $[\![\_,\_]\!]_{G}$ denotes Grabowski's bracket) 
\begin{equation*}
\func{d}[\![\alpha,\beta]\!]_{{G}}=[\![\func{d}\alpha ,\func{d}\beta]\!]_{{KS%
}}.
\end{equation*}
The problem with this bracket, is that it does not satisfy Leibniz's rule,
so it is not possible to speak about Hamiltonian vector fields (they are not
derivations). But following this, we also could say that $%
[\![\_,\_]\!]_{\Theta_{\omega,g}}$ is an integral of the Koszul-Schouten
bracket, in view of the theorem, as its associated symplectic form $%
\Theta_{\omega,g}$ gives through a (graded Lie) derivative the
Koszul-Schouten form $\Theta_{KS}$. Another bracket (which is neither graded
Poisson, but a Loday one) with similar properties is due to Y.
Kosmann-Schwarzbach (see \cite{[7]}).
\end{remark}

Let us see an interesting consequence of this result. In the odd graded
case, the fact that $\func{d}$ is a graded Hamiltonian vector field,
translates in a very useful relation: if $\alpha\in\Omega(M)$, then knowing $%
D_{\alpha}^{\Theta_{KS}}$ we also know $D_{\func{d}\alpha}^{\Theta_{KS}}$,
as the formula 
\begin{equation}
D_{\func{d}\alpha}^{\Theta_{KS}}=[\func{d},D_{\alpha}^{\Theta_{KS}}]
\label{eq4.1}
\end{equation}
holds. Indeed, this relation is simply expressing the fact that $\func{d}$
is a derivation of the Koszul-Schouten bracket, 
\begin{equation*}
\func{d}[\![\alpha,\beta]\!]_{\Theta_{KS}}-[\![\func{d}\alpha,\beta]\!]_{%
\Theta_{KS}}-[\![\alpha,\func{d}\beta]\!]_{\Theta _{KS}}=0.
\end{equation*}

In the even case, as already mentioned, $d$ can be never a graded
Hamiltonian vector field, so we expect a different relation. What is
remarkable, is that the failure of (\ref{eq4.1}) is given by the
Koszul-Schouten form in a direct way.

\begin{corollary}
Let $\alpha\in\Omega(M)$ be any differential form. Then 
\begin{equation}
D_{{\func{d}}\alpha}=[{\func{d}},D_{\alpha}]+(-1)^{\alpha
}\Theta_{\omega,g}^{-1}(\iota_{D_{\alpha}}\Theta_{KS}).  \label{eq4.2}
\end{equation}
\end{corollary}

\begin{proof}
Just compute, making use of some facts like $[\mathcal{L}_{D}^{G},\func{d}%
^{G}]=0$, for any derivation $D$ on $\Omega(M)$, and $[\mathcal{L}%
_{D}^{G},\iota_{E}]=\iota_{\lbrack D,E]}$ for any derivations $D,E$: 
\begin{align*}
\iota_{\lbrack\func{d},D_{\alpha}]}\Theta_{\omega,g} & =[\mathcal{L}_{\func{d%
}}^{G},\iota_{D_{\alpha}}]\Theta_{\omega ,g}=(\mathcal{L}_{\func{d}%
}^{G}\circ\iota_{D_{\alpha}}-(-1)^{\alpha }\iota_{D_{\alpha}}\circ\mathcal{L}%
_{\func{d}}^{G})\Theta_{\omega ,g}= \\
& =\mathcal{L}_{\func{d}}^{G}(\func{d}^{G}\alpha
)-(-1)^{\alpha}\iota_{D_{\alpha}}\Omega_{KS}=\func{d}^{G}\func{d}%
\alpha-(-1)^{\alpha}\iota_{D_{\alpha}}\Omega_{KS}= \\
& =\iota_{D_{\func{d}\alpha}}\Theta_{\omega,g}-(-1)^{\alpha}\iota_{D_{%
\alpha}}\Omega_{KS}
\end{align*}
\end{proof}

Translated in terms of graded Poisson brackets, this reads 
\begin{equation*}
\func{d}[\![\alpha,\beta]\!]_{\Theta_{\omega,g}}-[\![\func{d}%
\alpha,\beta]\!]_{\Theta_{\omega,g}}-[\![\alpha,\func{d}\beta]\!]_{\Theta_{%
\omega,g}}=\left\langle D_{\alpha},D_{\beta};\Theta _{KS}\right\rangle ,
\end{equation*}
so the Koszul-Schouten form appears as an object that forbids the
possibility for $\func{d}$ of being a graded Hamiltonian vector field with
respect to the class of forms $\Theta_{\omega,g}$.

Now, we can ask what condition on a general graded even symplectic form on $%
(M,\Omega(M))$ assures us a relation as in the theorem.

\begin{theorem}
With the notations of section $3$, 
\begin{equation*}
\iota_{\func{d}}\Theta_{\omega,g,L}=\lambda_{\omega}
\end{equation*}
if and only if $L\equiv0$. Thus, this property characterizes the subclass of
forms $\Theta_{\omega,g}$.
\end{theorem}

\begin{proof}
Recall that we can express the most general even symplectic form $\Theta$ as 
\begin{equation*}
\Theta=\Theta_{\omega,g,L}=\Theta_{\omega}+\frac{1}{2}\limfunc{d}%
\nolimits^{G}\lambda_{g}+\frac{1}{2}\limfunc{d}\nolimits^{G}\lambda_{L}
\end{equation*}
with $L\in\Gamma(TM)\otimes\Omega^{2}(M)$, where $\lambda_{L}$ is given by ($%
X\in\mathcal{X}(M)$) 
\begin{align*}
\left\langle \mathcal{L}_{X};\lambda_{L}\right\rangle & =L(X)\in\Omega
^{2}(M) \\
\left\langle i_{X};\lambda_{L}\right\rangle & =0,
\end{align*}
and 
\begin{align*}
\left\langle \mathcal{L}_{X};\lambda_{g}\right\rangle & =\limfunc{d}(g(X))=%
\limfunc{d}\flat(X) \\
\left\langle i_{X};\lambda_{g}\right\rangle & =g(X)=\flat(X).
\end{align*}
Now, we only have to compute the action of $\iota_{\limfunc{d}%
}\Theta_{\omega,g,L}\in\Omega_{G}^{1}$ on the basic derivations $i_{X},%
\mathcal{L}_{X}$. Thus, 
\begin{align*}
\left\langle i_{X},\limfunc{d};\Theta_{\omega,g,L}\right\rangle &
=\left\langle i_{X},\limfunc{d};\Theta_{\omega}\right\rangle +\frac
{1}{2}\left\langle i_{X},\limfunc{d};\limfunc{d}\nolimits^{G}\lambda_{g}%
\right\rangle +\frac{1}{2}\left\langle i_{X},\limfunc{d};\limfunc{d}%
\nolimits^{G}\lambda_{L}\right\rangle = \\
& =\frac{1}{2}i_{X}\left\langle \limfunc{d};\lambda_{L}\right\rangle -\frac{1%
}{2}L(X).
\end{align*}

To determine $i_{X}\left\langle \limfunc{d};\lambda_{L}\right\rangle $, we
proceed locally, and write $\limfunc{d}=\mathcal{L}_{Id}=\limfunc{d}%
x^{i}\wedge\emph{L}_{\partial_{x^{i}}}$; then, a direct calculation shows
that $i_{X}\left\langle \limfunc{d};\lambda _{L}\right\rangle =L(X)-\limfunc{%
d}x^{i}\wedge i_{X}L(\partial_{x^{i}}),$and so $\left\langle i_{X},\limfunc{d%
};\Theta_{\omega,g,L}\right\rangle =-\limfunc{d}x^{i}\wedge
i_{X}L(\partial_{x^{i}}).$ Evaluating on $Y,Z\in\mathcal{X}(M)$, we obtain 
\begin{equation*}
2\left\langle i_{X},\limfunc{d};\Theta_{\omega,g,L}\right\rangle
=L(Z;X,Y)-L(Y;X,Z).
\end{equation*}
Assume now that $\iota_{\limfunc{d}}\Theta_{\omega,g,L}=\lambda_{\omega }$.
Then, by the definition of $\lambda_{\omega}$, $L(Z;X,Y)-L(Y;X,Z)=0,$and
this, along with the antisymmetry in the two last arguments (recall $%
L(X;\_,\_)\in\Omega^{2}(M)$), implies that $L$ is symmetric in its two first
arguments. But, then 
\begin{equation*}
L(Y;Z,X)=L(Z;Y,X)=-L(Z;X,Y)=-L(X;Z,Y)=-L(Y;Z,X),
\end{equation*}
so $L=0$.

Reciprocally, if it is $L=0$, then $\Theta_{\omega,g,L}=\Theta_{\omega,g} $
and, as we have just seen, in this case $\iota_{\limfunc{d}%
}\Theta_{\omega,g,L}=\lambda_{\omega}$.
\end{proof}

\section{Computation of the graded hamiltonian vector fields and graded
Poisson brackets}

In the preceding section, we have given an expression for determining the
graded hamiltonian vector fields for differentials $d\alpha $, $D_{d\alpha }$
(see corollary $1$), when one knows the corresponding $D_{\alpha }$. But we
can improve this result and, in fact, we can derive explicit expressions for 
$D_{f}$ and $D_{df}$ when $f\in C^{\infty }(M)$; thus linearity will enable
us to compute any hamiltonian graded vector field.

We will begin by determining $D_{f}$. At first sight, it can be a
non-homogeneous derivation of the form $D_{f}=D_{f}^{-1}+\sum\limits_{i\in 
\mathbb{N}}D_{f}^{i}$ (each $D_{f}^{i}$ is an homogeneous term, with the
superindex denoting the degree), but it is actually simpler.

\begin{lemma}
The hamiltonian graded vector field $D_{f}$ only has terms in even degrees.
\end{lemma}

\begin{proof}
A straightforward checking of the $\mathbb{Z}_{2}-$degree at both members of
the equation $\iota_{D_{f}^{(even)}}\Omega_{\omega,g}+\iota_{D_{f}^{(odd)}}%
\Omega_{\omega,g}=\func{d}^{G}f$, implies that $D_{f}^{(odd)}=0$.
\end{proof}

As a consequence of this result, we will write 
\begin{equation*}
D_{f}=D_{f}^{(even)}=\sum_{i\in\mathbb{N}}D_{f}^{2i}.
\end{equation*}

\begin{remark}
Note that the sum is finite, the number of terms being bounded by $\dim M$.
\end{remark}

On the other hand, a theorem analogous to that of Fr\"{o}licher-Nijenhuis
(see \cite{[11]}), guarantees that given a connection $\nabla$ on $M$, each
homogeneous derivation $D_{f}^{2i}$ can be decomposed in the form 
\begin{equation}
D_{f}^{2i}=\nabla_{K_{f}^{2i}}+i_{L_{f}^{2i}}  \label{eq6.1}
\end{equation}
where $K_{f}^{2i}\in\Gamma(TM)\otimes\Omega^{2i}(M),L_{f}^{2i}\in
\Gamma(TM)\otimes\Omega^{2i+1}(M)$. What we are now going to prove, is that
the tensors $K_{f}^{2i}$ can be computed recursively, while the $L_{f}^{2i}$
are zero.

\begin{remark}
Having in mind $($\ref{eq6.1}$)$, we will write 
\begin{equation*}
K_{f}=\sum_{i\in\mathbb{N}}K_{f}^{2i}\text{ },\text{ }L_{f}=\sum _{i\in%
\mathbb{N}f}^{2i}L_{f}^{2i}
\end{equation*}
so that $D_{f}=\nabla_{K_{f}}+i_{L_{f}}$.
\end{remark}

\begin{lemma}
With the preceding notations, $L_{f}\equiv 0$.
\end{lemma}

\begin{proof}
Consider the graded $1$-form $\func{d}^{G}f=\iota_{D_{f}}\Omega_{\omega,g} $%
, and make it act upon the derivation $i_{X}\in \mathcal{D}_{1}^{-1}$, with $%
X\in\mathcal{X}(M)$.
\end{proof}

\begin{proposition}
With the preceding notations 
\begin{equation*}
K_{f}^{2i}=(-1)^{i}J^{-1}(R(\_,\_)(J^{-1}(.\overset{i}{\breve{.}}%
.J^{-1}(R(\_,\_)X_{f})...))),
\end{equation*}
where $J\in T^{(1,1)}(M)$ is the tensor field determined by $\omega
(X,Y)=g(JX,Y)$, $R\in \Omega ^{2}(M;TM)$ is the curvature of $g$ and $X_{f}$
is the hamiltonian vector field of $f$ with respect to $\omega $. In
particular, we have the recursion formula 
\begin{equation*}
K_{f}^{2i}=-J^{-1}(R(\_,\_)K_{f}^{2(i-1)})
\end{equation*}
\end{proposition}

\begin{proof}
This time, we study the action of the graded $1$-form $\func{d}%
^{G}f=\iota_{D_{f}}\Omega_{\omega,g}$ on the basic derivations of the type $%
\nabla_{Y}$, $Y\in\mathcal{X}(M)$: 
\begin{equation*}
Y(f)=\left\langle \nabla_{Y};\iota_{D_{f}}\Omega_{\omega,g}\right\rangle
=\omega(Y,K_{f}^{0}+K_{f}^{2}(\_,\_)+...)+g(R(Y,K_{f}^{0}+K_{f}^{2}(\_,%
\_)+...)\_,\_).
\end{equation*}

Note that the right member belongs to $\Omega^{0+2+...}(M).$ Now, equate
terms with the same degree. For degree $0$ ,we have 
\begin{equation*}
Y(f)=\omega(Y,K_{f}^{0})=\func{d}f(Y)=i_{X_{f}}\omega(Y)
\end{equation*}
so 
\begin{equation*}
K_{f}^{0}=X_{f}\in\mathcal{X}(M).
\end{equation*}

Now, the degree $2$ terms give 
\begin{align*}
0 & =\omega(Y,K_{f}^{2}(\_,\_))+g(R(Y,K_{f}^{0})\_,\_)= \\
& =-g(Y,J(K_{f}^{2}(\_,\_)))+R(\_,\_,X_{f},Y)= \\
& =-g(Y,J(K_{f}^{2}(\_,\_)))-g(Y,R(\_,\_)X_{f})
\end{align*}
thus 
\begin{equation*}
J(K_{f}^{2}(\_,\_))=-R(\_,\_)X_{f}
\end{equation*}
and 
\begin{equation*}
K_{f}^{2}=-J^{-1}(R(\_,\_)X_{f}),
\end{equation*}
(we have made use of the Riemann curvature tensor $R(U,V,W,Z)=-g(R(U,V)W,Z)$
and its symmetries).

Iterating this procedure, for the terms of degree $2i$ we get 
\begin{align*}
0 & =\omega(Y,K_{f}^{2i}(\_,...,\_))+g(R(Y,K_{f}^{2(i-1)}(\_,...,\_))\_,\_)=
\\
& =-g(Y,J(K_{f}^{2i}(\_,...,\_)))-g(R(\_,\_)K_{f}^{2(i-1)}(\_,...,\_),Y)
\end{align*}
\newline
and the claim follows.
\end{proof}

Thus, we have proved the following result.

\begin{theorem}
If we write $K_{f}^{(even)}=\sum\limits_{i\in \mathbb{N}}K_{f}^{2i}$ (where $%
K_{f}^{2i}$ is defined in Proposition$17$), then $D_{f}=\nabla
_{K_{f}^{(even)}}.$
\end{theorem}

Note also that we have obtained in the proof: 
\begin{equation*}
K_{f}^{0}=X_{f}\text{.}
\end{equation*}

Now, the same analysis can be applied to find the graded hamiltonian vector
field of a differential $df,f\in C^{\infty }(M)$. It must verify the
equation 
\begin{equation*}
\iota _{D_{df}}\Omega _{\omega ,g}=\func{d}^{G}df,
\end{equation*}
and so it must have the form 
\begin{equation*}
D_{df}=i_{\sum\limits_{i\in \mathbb{N}}N_{f}^{2i}}+\nabla
_{\sum\limits_{i\in \mathbb{N}}K_{f}^{2i+1}},
\end{equation*}
where $N_{f}^{2i}\in \Gamma (TM)\otimes \Omega ^{2i}(M),$\bigskip $%
K_{f}^{2i+1}\in \Gamma (TM)\otimes \Omega ^{2i+1}(M)$. However, as for the
case of $D_{f}$, this formula can be simplified a lot with some computations
similar to those of the preceding Proposition.

\begin{theorem}
With the preceding notations, 
\begin{equation*}
D_{df}=i_{\natural df}+\nabla _{\sum\limits_{i\in \mathbb{N}%
}K_{f}^{2i+1}}=i_{\natural df}+\nabla _{K_{f}^{(odd)}},
\end{equation*}
where $\natural \equiv g^{-1}$ is the musical isomorphism induced by the
metric.
\end{theorem}

Now, with the same techniques as in the case of the $K_{f}^{2i}$'s, we can
obtain a recursion formula for the tensors $K_{f}^{2i+1}$.

\begin{proposition}
With the preceding notations, we have 
\begin{equation*}
K_{f}^{2(i+1)+1}=J^{-1}(R(\_,\_)K_{f}^{2i+1}),\text{ }for\text{ }i>0
\end{equation*}
and 
\begin{equation*}
\omega (X,K_{f}^{1})=\nabla _{X}(df).
\end{equation*}
\end{proposition}

\begin{remark}
The formulae we have obtained, are general and make no assumptions on $%
g,\omega$ or $J$. But for the cases of interest, some compatibility relation
between $g$ and $\omega$ exists, which translates into a condition on $J$;
the most common one, is to consider K\"{a}hler manifolds, and then the
equivalent conditions $\nabla\omega=0$, $\nabla J=0$, are satisfied. From
now on, we will take $(M,g,\omega,J)$ a K\"{a}hler manifold, unless
otherwise stated.
\end{remark}

We know that the recursion formula for the tensors $K_{f}^{2i}$ begins with $%
K_{f}^{0}=X_{f}$. Now, under the hypothesis of K\"{a}hler manifolds, we will
deduce the first term of the $K_{f}^{2i+1}$'s.

\begin{remark}
We will make use of the definition and properties of the exterior covariant
derivative $d^{\nabla}$, which is a derivation of the algebra of
vector-valued differential forms (see \cite{[9]}).
\end{remark}

\begin{lemma}
With the preceding notations 
\begin{equation*}
K_{f}^{1}=-d^{\nabla }X_{f}.
\end{equation*}
\end{lemma}

\begin{proof}
The tensor $K_{f}^{1}$ is characterized by the condition 
\begin{equation*}
\omega(Y,K_{f}^{1})=\nabla_{Y}(df),\text{ }Y\in\mathcal{X}(M),
\end{equation*}
and this is an identity between $1-$forms, so we make it act upon a vector
field $U\in\mathcal{X}(M)$: 
\begin{align*}
\omega(Y,K_{f}^{1}(U)) & =(\nabla_{Y}(df))(U)= \\
& =\nabla_{Y}(df(U))-df(\nabla_{Y}(U))= \\
& =Y(U(f))-(\nabla_{Y}U)(f).
\end{align*}

On the other hand: 
\begin{equation*}
\omega(Y,-(d^{\nabla}X_{f})(U))=\omega(Y,\nabla_{U}X_{f}),
\end{equation*}
and with $M$ K\"{a}hler ($d^{\nabla}\omega=0$): 
\begin{align*}
-\omega(Y,\nabla_{U}X_{f}) & =-\nabla_{U}(\omega(Y,X_{f}))+\omega(\nabla
_{U}Y,X_{f})= \\
& =U(Y(f))-(\nabla_{U}Y)(f).
\end{align*}

Thus, 
\begin{equation*}
\omega(Y,K_{f}^{1}(U))-\omega(Y,-(d^{\nabla}X_{f})(U))=0,
\end{equation*}
and, from the non degeneracy of $\omega$, $K_{f}^{1}=-d^{\nabla}X_{f}.$
\end{proof}

This result could also be written 
\begin{equation*}
K_{f}^{1}=-d^{\nabla}K_{f}^{0},
\end{equation*}
and so we can ask if a similar relation holds between the tensors $%
K_{f}^{2i+1}$ and $K_{f}^{2i}$ for any $i>0$.

\begin{proposition}
With the preceding notations, 
\begin{equation*}
K_{f}^{2i+1}=(-1)^{i+1}d^{\nabla }K_{f}^{2i}.
\end{equation*}
\end{proposition}

\begin{proof}
We will proceed by induction (the case $i=0$ has been just proved). Assuming
that 
\begin{equation*}
K_{f}^{2i+1}=(-1)^{i+1}d^{\nabla}K_{f}^{2i}
\end{equation*}
we have, from the recursion formula for the $K_{f}^{2i+1}$'s: 
\begin{equation*}
K_{f}^{2(i+1)+1}=J^{-1}(R(\_,\_)K_{f}^{2i+1})=(-1)^{i+1}J^{-1}(R(\_,\_)d^{%
\nabla}K_{f}^{2i}).
\end{equation*}

Now, by the properties of a K\"{a}hler manifold ($d^{\nabla }J=0=d^{\nabla
}R $), 
\begin{align*}
K_{f}^{2(i+1)+1}& =(-1)^{i+1}d^{\nabla }(J^{-1}(R(\_,\_)K_{f}^{2i}))= \\
& =(-1)^{i+1}d^{\nabla }(-K_{f}^{2(i+1)})= \\
& =(-1)^{(i+1)+1}d^{\nabla }K_{f}^{2(i+1)}
\end{align*}
\end{proof}

\subsection{Explicit expressions for the Poisson brackets}

Now we can give the general explicit expressions for the Poisson brackets
corresponding to the graded symplectic form $\Omega_{\omega,g}$. The
computations leading to these are quite tedious, but straightforward, so we
will omite them.

If $\alpha,\beta\in\Omega(M)$ are differential forms, their Poisson bracket
induced by $\Omega_{\omega,g}$ is denoted by $[\![\alpha,\beta]\!]_{\Omega
_{\omega,g}}$. By linearity and Leibniz property, in order to compute it we
only need the bracket between functions and/or exact $1-$forms, and the
formulae are ($f,h\in C^{\infty}(M)$): 
\begin{align*}
\lbrack\![f,h]\!]_{\Omega_{\omega,g}} & =\{f,h\}+\sum\limits_{i\in\mathbb{N}%
}R(K_{f}^{2i},X_{h},\_,\_) \\
\lbrack\![f,dh]\!]_{\Omega_{\omega,g}} & =d\{f,h\}-\sum\limits_{i\in \mathbb{%
N}}\omega(d^{\nabla}K_{f}^{2i},X_{h}) \\
& +\sum\limits_{i\in\mathbb{N}}(R(d^{\nabla}K_{f}^{2i},X_{h},\_,%
\_)+R(K_{f}^{2i},d^{\nabla}X_{h},\_,\_)) \\
\lbrack\![df,dh]\!]_{\Omega_{\omega,g}} & =g^{-1}(df,dh) \\
& +g(d^{\nabla}\#df,d^{\nabla}X_{h})+R(\#df,X_{h},\_,\_)+\sum\limits_{i\in 
\mathbb{N}}\omega(X_{h},d^{\nabla}K_{f}^{2i+1}) \\
& +\sum\limits_{i\in\mathbb{N}}(R(d^{\nabla}K_{f}^{2i+1},X_{h},\_,%
\_)+R(K_{f}^{2i+1},d^{\nabla}X_{h},\_,\_)).
\end{align*}

Here, $\{\_,\_\}$ is the Poisson bracket associated to $\omega\in\Omega
^{2}(M)$.

\section{The graded Hamiltonian vector field associated to a symplectic form
on the base manifold}

In this section we want to compute a special graded Hamiltonian vector
field. If we have a base manifold $(M,\omega,g)$, a distinguished graded
Hamiltonian function is $\omega\in\Omega^{2}(M)$, so it is natural to ask
what is its corresponding graded Hamiltonian vector field. We know that in
the odd case, this is nothing but $\limfunc{d}$, and also we know that this
is not possible in the even one. In the introduction, we have said that $%
i_{J}$ (with $J$ given by $\omega(X,Y)=g(JX,Y)$) plays the r\^{o}le of $%
\limfunc{d}$, and we will prove now this fact completing the table appearing
there.

\begin{lemma}
\label{lema} Let $g$ be a pseudoriemannian metric on a differentiable
manifold $M$ and let $J\in\Omega^{1}(M;TM)$ such that 
\begin{equation*}
g(JX,Y)=-g(X,JY),
\end{equation*}
then, for all $X,Y,Z\in\mathcal{X}(M)$. 
\begin{equation*}
g((\nabla_{X}J)Y,Z)=g((\nabla_{X}J)Z,Y),
\end{equation*}
where $\nabla$ is the Levi-Civita connection associated to $g$.
\end{lemma}

\begin{proof}
A straightforward computation.
\end{proof}

\begin{theorem}
Let $\omega$ be a symplectic form and $g$ a pseudoriemannian metric on a
differentiable manifold, $M$. Let $\Theta_{\omega,g}$ be the associated even
symplectic form. The graded Hamiltonian vector field associated to the
graded function $\omega\in\Omega^{2}(M)$ is $D_{\omega}=i_{J}$ with $%
J\in\Omega ^{1}(M;TM)$ defined by $\omega(X,Y)=g(JX,Y)$ for all $X,Y\in%
\mathcal{X}(M)$.
\end{theorem}

\begin{proof}
Let us check that, as graded $1-$forms, 
\begin{equation*}
\iota_{i_{J}}\Theta_{\omega,g}=\limfunc{d}\nolimits^{G}\omega.
\end{equation*}

Given a vector field $Y$, thanks to Proposition $6$, 
\begin{align*}
\left\langle i_{J},i_{Y};\Theta_{\omega,g}\right\rangle & =g(J,Y)=-\omega
(\_,Y)= \\
& =-\omega(Y,\_)=-i_{Y}\omega=-\left\langle i_{Y};\limfunc{d}%
\nolimits^{G}\omega\right\rangle .
\end{align*}

before continuing with the computations, note that a consequence of the
definition of $J$ is that 
\begin{equation}
g(JX,Y)=-g(X,JY),  \label{eq5.1}
\end{equation}
for all $X,Y\in\mathcal{X}(M)$. Moreover, this implies that 
\begin{equation}
\iota_{i_{J}}\lambda_{g}=2\omega.  \label{eq5.2}
\end{equation}

Indeed, $(\iota_{i_{J}}\lambda_{g})(X,Y)=<i_{J};%
\lambda_{g}>(X,Y)=g(JX,Y)-g(JY,X)=2\omega(X,Y)$.

Now, for any vector field $Y$, by the definition of the graded even
symplectic form and the graded exterior derivative, 
\begin{align*}
\left\langle \mathcal{L}_{Y},i_{J};\Theta_{\omega,g}\right\rangle &
=\left\langle \mathcal{L}_{Y},i_{J};\frac{1}{2}\limfunc{d}%
\nolimits^{G}\lambda_{g}\right\rangle = \\
& =\frac{1}{2}(\mathcal{L}{}_{Y}<i_{J};\lambda_{g}>-i_{J}<\mathcal{L}%
{}_{Y};\lambda_{g}>-<[{}\mathcal{L}_{Y},i_{J}];\lambda_{g}>)= \\
& =\mathcal{L}{}_{Y}\omega-\frac{1}{2}i_{J}(g(Y,\_))-\frac{1}{2}<i_{{}%
\mathcal{L}_{Y}J};\lambda_{g}>,
\end{align*}
where we have applied (\ref{eq3.1}), (\ref{eq5.2}) and the formula for the
commutator of the derivations ${}\mathcal{L}_{Y}$ and $i_{J}$.

A long but straightforward computation using Lemma $13$, by virtue of (\ref
{eq5.1}), shows that the terms $i_{J}(g(Y,\_))+<i_{{}\mathcal{L}%
_{Y}J};\lambda_{g}>$ vanish. Therefore, for all $Y\in\mathcal{X}(M)$, 
\begin{align*}
\left\langle i_{Y},i_{J};\Theta_{\omega,g}\right\rangle & =\left\langle
i_{Y};\limfunc{d}\nolimits^{G}\omega\right\rangle \\
\left\langle \mathcal{L}_{Y},i_{J};\Theta_{\omega,g}\right\rangle &
=\left\langle \mathcal{L}_{Y};\limfunc{d}\nolimits^{G}\omega \right\rangle .
\end{align*}

This implies that for all derivations $D\in\limfunc{Der}\Omega (M)$,\newline
$<D,i_{J};\Theta_{\omega,g}>=<D;\limfunc{d}{^{G}}\omega>,$ i.e, $%
\iota_{i_{J}}\Theta_{\omega,g}=\limfunc{d}{^{G}}\omega.$
\end{proof}

We have just seen that a subset of the even graded symplectic forms of $%
\mathbb{Z}-$degree $(0)+(2)$ has the property that $i_{J}$ is a Hamiltonian
vector field. Now, we give a result telling us how the most general forms in
which this property remains true are.

\begin{theorem}
Let $\Theta=\Theta_{\omega}+\frac{1}{2}\func{d}^{G}\lambda_{g,L}$ be an
arbitrary even symplectic form of $\mathbb{Z}-$degree $(0)+(2)$. Then, $%
i_{J} $ is a locally Hamiltonian graded vector field for $\Theta$ if and
only if $\ L$ is such that $L(X;JY,Z)=-L(X;Y,JZ)$ for all $X,Y,Z\in\mathcal{X%
}(M)$, where $J\in\Omega^{1}(M;TM)$ defined by $\omega(X,Y)=g(JX,Y)$.
\end{theorem}

\begin{proof}
Recall that, as we mentioned at the end of the preceding section, any even
symplectic form of $\mathbb{Z}-$degree $(0)+(2)$ is uniquely determined by a
symplectic form $\omega$, a pseudoriemannian metric $g$, and a tensor field $%
L:TM\longrightarrow\Lambda^{2}T^{\ast}M$. For the particular case we are
considering, we have 
\begin{equation}
\Theta=\Theta_{\omega,g,L}=\Theta_{\omega}+\frac{1}{2}\limfunc{d}{^{G}}%
\lambda_{(g,L)}=\Theta_{\omega,g}+\frac{1}{2}\limfunc{d}{^{G}}%
\lambda_{(0,L)}.
\end{equation}

If $i_{J}$ is a locally Hamiltonian vector field, then 
\begin{align}
0 & =\mathcal{L}_{i_{J}}^{G}\Theta=\limfunc{d}{^{G}}\iota_{i_{J}}\Theta=%
\limfunc{d}{^{G}}(\iota_{i_{J}}\Theta_{\omega,g}+\frac{1}{2}\iota_{i_{J}}%
\limfunc{d}{^{G}}\lambda_{(0,L)})=  \label{eq5.4} \\
& =\limfunc{d}{^{G}}(\limfunc{d}{^{G}}\omega+\frac{1}{2}\iota_{i_{J}}%
\limfunc{d}{^{G}}\lambda_{(0,L)})=\frac{1}{2}\limfunc{d}{^{G}}\iota_{i_{J}}%
\limfunc{d}{^{G}}\lambda _{(0,L)}.  \notag
\end{align}

But note that the closed graded $1-$form $\iota_{i_{J}}\limfunc{d}{^{G}}%
\lambda_{(0,L)}$ has $\mathbb{Z}-$degree $2$. Therefore it is exact, i.e.,
there exists $\alpha\in\Omega^{2}(M)$ (including the factor $\frac12$) such
that 
\begin{equation*}
\iota_{i_{J}}\limfunc{d}{^{G}}\lambda_{(0,L)}=\limfunc{d}{^{G}}\alpha.
\end{equation*}
This implies that $i_{J}$ is globally Hamiltonian since $\iota_{i_{J}}\Theta=%
\limfunc{d}{^{G}}(\omega+\alpha)$. Now, as $\iota_{i_{J}}\Omega=\limfunc{d}{%
^{G}}\omega=\iota_{i_{J}}\Theta_{\omega,g}$, then $\iota_{i_{J}}\limfunc{d}{%
^{G}}\lambda_{(0,L)}=0$. We shall prove next that $\alpha=0$. Indeed, since 
\begin{equation}
<i_{X},i_{J};\limfunc{d}{^{G}}\lambda_{(0,L)}>=0  \label{eq5.5}
\end{equation}
we get $0=i_{X}\alpha$ for all $X$. Then $\alpha=0$. Finally, we shall prove
that $L$ verifies the required condition by a direct calculation. For any $%
Y,Z\in\mathcal{X}(M),$%
\begin{align}
0 & =\left\langle \mathcal{L}_{X},i_{J};\limfunc{d}{^{G}}\lambda_{(0,L)}%
\right\rangle (Y,Z)  \label{eq5.6} \\
& =(-i_{J}(L(X))-\left\langle i_{[X,J]};\lambda_{(0,L)}\right\rangle )(Y,Z) 
\notag \\
& =-i_{J}(L(X))(Y,Z)=-L(X;JY,Z)-L(X;Y,JZ).  \notag
\end{align}
Therefore $L(X;JY,Z)=-L(X;Y,JZ)$.

For the converse, assume now this condition on $L$. As before, 
\begin{align*}
\mathcal{L}_{i_{J}}^{G}\Theta & =\limfunc{d}{^{G}}\iota_{i_{J}}\Theta=%
\limfunc{d}{^{G}}(\iota_{i_{J}}\Theta_{\omega,g}+\iota_{i_{J}}\limfunc{d}{%
^{G}}\lambda_{(0,L)})= \\
& =\limfunc{d}{^{G}}\limfunc{d}{^{G}}\omega+\limfunc{d}{^{G}}\iota_{i_{J}}%
\limfunc{d}{^{G}}\lambda_{(0,L)}= \\
& =\limfunc{d}{^{G}}\iota_{i_{J}}\limfunc{d}{^{G}}\lambda_{(0,L)}
\end{align*}

But from (\ref{eq5.5}) and (\ref{eq5.6}), we have $\iota_{i_{J}}\limfunc{d}{%
^{G}}\lambda_{(0,L)}=0$, so $\mathcal{L}_{i_{J}}^{G}\Omega=0$.
\end{proof}

Thus, in order to be $i_{J}$ Hamiltonian, $L(X;\_,\_)$ must present the same
property as $g$ in Lemma \ref{lema}.

\section{ Some applications and examples}

In the preceding sections, we have obtained the necessary and sufficient
condition that a $(1,1)$-tensor field $J$ must fulfil in order to induce a
derivation $i_{J}$ which is a graded Hamiltonian vector field with respect
to the graded symplectic form $\Omega _{\omega ,g}$. Now, we would like to
present some examples of this situation. The condition on $J$, is expressed
in the form of a relation between $J,g$ and $\omega $ similar to that
existing in the field of complex differential geometry: 
\begin{equation*}
\omega (X,Y)=g(JX,Y)\qquad \forall X,Y\in \mathcal{X}(M),
\end{equation*}
but this relation is also characteristic of another kind of geometry (maybe
not so widespread as the complex one), called paracomplex geometry, a survey
of which can be consulted in \cite{[3]}, where numerous applications to
subjects as Finsler spaces, quantizable coadjoint orbits or negatively
curved manifolds are referenced. Indeed, there is a close parallelism
between both geometries, and the main difference is the replacement of the
restriction on $J$ to be an almost complex structure (in the case of complex
geometry) by demanding $J$ to be an almost product structure (in the case of
the paracomplex one).

For the sake of completeness and comparison, we recall now the basic
definitions of complex and paracomplex geometry.

\begin{definition}
Let $M$ be a manifold and $J\in T^{(1,1)}M$. We say that

\begin{enumerate}
\item  $J$ is an almost product structure, if $J^{2}= Id$,

\item  $J$ is an almost complex structure, if $J^{2}=-Id$.
\end{enumerate}
\end{definition}

\begin{definition}
An almost para-Hermitian manifold $(M,g,J)$, is a manifold $M$ endowed with
an almost product structure $J$ and a pseudoriemannian metric $g$,
compatible in the sense that 
\begin{equation}
g(JX,JY)=-g(X,Y)\text{ \ \ \ , \ }\forall X,Y\in\mathcal{X}(M).
\label{paraher}
\end{equation}
\end{definition}

\begin{definition}
An almost Hermitian manifold $(M,h,K)$, is a manifold $M$ endowed with an
almost complex structure $K$ and a riemannian metric $g$, compatible in the
sense that 
\begin{equation}
h(KX,KY)=h(X,Y)\text{ \ \ \ , \ \ }\forall X,Y\in\mathcal{X}(M).
\label{hermit}
\end{equation}
\end{definition}

\begin{remark}
Note that, in any case, the hypothesis of lemma \ref{lema} is satisfied.
\end{remark}

For an almost para-Hermitian manifold $(M,g,J)$, we can define its
fundamental $2-$form $\alpha$ as 
\begin{equation*}
\alpha(X,Y)=g(JX,Y),
\end{equation*}
and similarly, for an almost-Hermitian manifold $(M,h,K)$, its fundamental $%
2-$form $\beta$ is 
\begin{equation*}
\beta(X,Y)=h(KX,Y).
\end{equation*}

\begin{definition}
An almost para-K\"{a}hler manifold is an almost para-Hermitian manifold $%
(M,g,J)$ such that its fundamental $2-$form $\alpha$ is closed, that is $%
d\alpha=0.$
\end{definition}

\begin{definition}
An almost K\"{a}hler manifold is an almost Hermitian manifold $(M,h,K)$ such
that its fundamental $2-$form $\beta$ is closed, that is $d\beta=0.$
\end{definition}

\begin{remark}
In any of the preceding definitions, the prefix ``almost'' can be dropped if
an integrability condition is imposed on $J$ (resp. $K$).
\end{remark}

Now, as a direct consequence of our previous work, we can give a
charaterization of these structures in terms of graded Hamiltonian vector
fields.

\begin{theorem}
Let $(M,f,L)$ be either an almost para-Hermitian or an almost Hermitian
manifold, with $\gamma\in\Omega^{2}(M)$ the corresponding fundamental $2-$%
form. Then, $L\in T^{(1,1)}M$ is the unique tensor field for which $i_{L}$
is the graded Hamiltonian vector field associated to $\gamma$ with respect
to the graded symplectic form $\Omega_{\gamma,f}$.
\end{theorem}

\begin{proof}
This result is just a corollary to Theorem $13$.
\end{proof}

\begin{remark}
In particular, this result applies to (almost) para-K\"{a}hler and (almost)
K\"{a}hler manifolds.
\end{remark}

In order to get a feeling of how paracomplex geometry enter into Mechanics,
let us consider a very simple example. Later, we will show how it can be
used to give a characterization of some canonical structures on the tangent
space of a manifold.

Let $(M,g)$ be a pseudoriemannian manifold, and let $L\in C^{\infty}(TM)$ be
the Lagrangian defined by the quadratic form associated to $g$; in the local
coordinates $\{q^{i},v^{i}\}_{1\leq i\leq n}$ of $TM$ ($n$ being the
dimension of $M)$, it is written as $L(v)=\frac{1}{2}g_{ij}(q)v^{i}v^{j},v%
\in T_{q}M.$

It is well known that $L$ generates a symplectic form on $TM$, denoted by $%
\omega_{L}$, which in coordinates has the expression 
\begin{equation*}
\omega_{L}=d(\frac{\partial L}{\partial v^{i}}dq^{i})=\frac{\partial^{2}L}{%
\partial q^{j}\partial v^{i}}dq^{j}\wedge dq^{i}+\frac{\partial^{2}L}{%
\partial v^{j}\partial v^{i}}dv^{j}\wedge dq^{i},
\end{equation*}
so, for the particular Lagrangian we are considering, it reduces to 
\begin{equation}
\omega_{L}={\partial_{k}g_{il}}dq^{j}\wedge dq^{i}+g_{ij}dv^{j}\wedge dq^{i}.
\label{lagran}
\end{equation}

Now, we will employ the theory of horizontal lifts of tensor fields to the
tangent space of a manifold, which can be consulted in \cite{[14]}. The
lifting is understood to be associated to a connection $\bar{\nabla}$ on $M$%
; let $\bar{\Gamma}_{jk}^{i}$ be its Christoffel symbols in the coordinates $%
(q^{i})_{1\leq i\leq n}$. Then, a vector field $X\in \mathcal{X}(M)$ gives
rise to its horizontal and vertical lifts with respect to $\bar{\nabla}%
,X^{h},X^{v}$, and a canonical almost product on $TM$ is defined by putting 
\begin{equation*}
\begin{array}{rcl}
J(X^{v}) & := & X^{v} \\ 
J(X^{h}) & := & -X^{h}.
\end{array}
\end{equation*}
\qquad\ \ This can be written matricially as 
\begin{equation}
J\equiv 
\begin{pmatrix}
1 & 0 \\ 
0 & -1
\end{pmatrix}
\label{defJ}
\end{equation}
The dual coframe adapted to $\bar{\nabla}$ is $\{\theta ^{i},\eta
^{i}\}_{i=1}^{n}$, where 
\begin{equation*}
\begin{array}{rcl}
\theta ^{i} & := & dq^{i} \\ 
\eta ^{i} & := & dv^{i}+\bar{\Gamma}_{jk}^{i}v^{j}dq^{k}.
\end{array}
\end{equation*}
We can define the horizontal lift to $TM$ of a $(0,2)$-tensor field $S$ as $%
S_{\bar{\nabla}}^{H}$, through 
\begin{align*}
S_{\bar{\nabla}}^{H}(X^{v},Y^{v})& =0 \\
S_{\bar{\nabla}}^{H}(X^{v},Y^{h})& =S_{\bar{\nabla}%
}^{H}(X^{h},Y^{v})=(S(X,Y))^{v} \\
S_{\bar{\nabla}}^{H}(X^{h},Y^{h})& =0,
\end{align*}
for all $X,Y\in \mathcal{X}(M)$, so if $S=S_{ij}dq^{i}\otimes dq^{j}$, then
in the adapted coframe 
\begin{equation*}
S_{\bar{\nabla}}^{H}=S_{ij}\theta ^{i}\otimes \eta ^{j}+S_{ij}\eta
^{i}\otimes \theta ^{j}.
\end{equation*}

Consider now the particular case of the symmetric connection ${\bar{\nabla}}$
defined locally from $g$ (and thus associated to $L$) as $\bar{\Gamma}%
_{lk}^{j}=\frac{1}{2}g^{jm}\partial_{m}g_{lk}.$ Then, we can lift the metric 
$g$ to a metric $g_{\overline{\nabla}}^{H}$ in $TM$ with the expression $g_{%
\bar{\nabla}}^{H}=g_{ij}\theta^{i}\otimes\eta^{j}+g_{ij}\eta^{i}\otimes%
\theta^{j}$, or, in matrix form, 
\begin{equation}
g_{\bar{\nabla}}^{H}\equiv 
\begin{pmatrix}
0 & g \\ 
g & 0
\end{pmatrix}
.  \label{defg}
\end{equation}

\begin{remark}
In the coordinates $\{q^{i},v^{i}\}_{1\leq i\leq n}$, we have 
\begin{equation*}
g_{\bar{\nabla}}^{H}=(\partial_{k}g_{ij})v^{l}dq^{i}\otimes
dq^{k}+2g_{ij}dq^{i}\otimes dv^{j}.
\end{equation*}
\end{remark}

On the other hand, from (\ref{lagran}) it is easy to see that $\omega_{L}$
can be represented by the matrix 
\begin{equation}
\omega_{L}\equiv 
\begin{pmatrix}
0 & g \\ 
-g & 0
\end{pmatrix}
\label{defomega}
\end{equation}
corresponding to the expression in the adapted coframe $\omega_{L}=g_{ij}%
\theta^{i}\otimes\eta^{j}-g_{ij}\eta^{i}\otimes\theta^{j}.$

But now, from (\ref{defJ}), (\ref{defg}) and (\ref{defomega}) we have 
\begin{equation}
\omega_{L}(A,B)=g_{\bar{\nabla}}^{H}(JA,B),\qquad\forall A,B\in\mathcal{X}%
(TM).  \label{formula5}
\end{equation}

In this expression, $\omega_{L}$ is a symplectic form and $g_{\bar{\nabla}%
}^{H}$ is a metric, so it resembles the characteristic relation of the
K\"{a}hler geometry, except for the fact that here $J^{2}=Id$, and not $%
J^{2}=-Id$. This is precisely the main feature of paracomplex geometry, as
we have seen. So, $(TM,g_{\bar{\nabla}}^{H},J,\omega_{L})$ is a
para-K\"{a}hler manifold.

\begin{remark}
A whole class of examples of this kind, treated in the spirit of Finsler
geometry, can be consulted in \cite{[12]}.
\end{remark}

We finish with an application of this example from Mechanics to Geometry.

\begin{proposition}
The canonical almost product structure $J$ on $TM$ is such that $i_{J}$ is
the unique graded Hamiltonian vector field for $\omega_{L}$ with respect to $%
\Omega_{\omega_{L},g_{\bar{\nabla}}^{H}}$, this being independent of the
pseudoriemannian metric $g$ chosen.
\end{proposition}

\begin{proof}
This is just a corollary of Theorem $14$.
\end{proof}

{\bf Acknowledgements:} The first author (J.M.) wants to acknowledge Y.
Kosmann-Schwarzbach for her encouragement after reading a first version of
this paper. The second author (J.A.V) wants to express his gratitude to the
CIMAT, Guanajuato (M\'{e}xico), for its warm hospitality during his stay
there, when part of this work was done. Special thanks are due to Adolfo S%
\'{a}nchez-Valenzuela and Mitchell Rothstein for very helpful discussions.

Work partially supported by a Grant from the Spanish Ministerio de
Educaci\'{o}n y Cultura, ref. PB-97-1386.

\end{document}